%% file: aaai24.tex
\documentclass[letterpaper]{article} % DO NOT CHANGE THIS
\usepackage{aaai24}  % DO NOT CHANGE THIS
\usepackage{times}  % DO NOT CHANGE THIS
\usepackage{helvet}  % DO NOT CHANGE THIS
\usepackage{courier}  % DO NOT CHANGE THIS
\usepackage[hyphens]{url}  % DO NOT CHANGE THIS
\usepackage{graphicx} % DO NOT CHANGE THIS
\usepackage{natbib}  % DO NOT CHANGE THIS AND DO NOT ADD ANY OPTIONS TO IT
\usepackage{caption} % DO NOT CHANGE THIS AND DO NOT ADD ANY OPTIONS TO IT
\usepackage[ruled,vlined]{algorithm2e}

\usepackage{color}
\usepackage{amsmath}
\usepackage{subfigure}
\usepackage{multirow}
\usepackage{diagbox}
\usepackage{booktabs}
\usepackage{arydshln}

\usepackage{newfloat}
\usepackage{listings}
\DeclareCaptionStyle{ruled}{labelfont=normalfont,labelsep=colon,strut=off} % DO NOT CHANGE THIS
\usepackage{xspace}
\newcommand{\name}[0]{KD-Club\xspace}
\makeatletter
\DeclareRobustCommand\onedot{\futurelet\@let@token\@onedot}
\def\@onedot{\ifx\@let@token.\else.\null\fi\xspace}

\def\eg{\emph{e.g}\onedot} 
\def\ie{\emph{i.e}\onedot}

\makeatother

\usepackage{xcolor}

\title{KD-Club: An Efficient Exact Algorithm with New Coloring-based Upper Bound\\for the Maximum $k$-Defective Clique Problem}

\author{
    Mingming Jin\equalcontrib, %\dag{The first two authors contribute equally.},  % mingmingk@hust.edu.cn
    Jiongzhi Zheng\equalcontrib,    % jzzheng@hust.edu.cn
     Kun He\thanks{Corresponding author.}
    \\
}
\affiliations{
    School of Computer Science and Technology, Huazhong University of Science and Technology, Wuhan, China\\
    brooklet60@hust.edu.cn
}

\usepackage{bibentry}
\usepackage{amsthm}

\newtheorem{lemma}{Lemma}

\begin{document}

\maketitle

\begin{abstract}
The Maximum $k$-Defective Clique Problem (MDCP) aims to find a maximum $k$-defective clique in a given graph, where a $k$-defective clique is a relaxation clique missing at most $k$ edges. MDCP is NP-hard and finds many real-world applications in analyzing dense but not necessarily complete subgraphs. Exact algorithms for MDCP mainly follow the Branch-and-bound (BnB) framework, whose performance heavily depends on the quality of the upper bound on the cardinality of a maximum $k$-defective clique. The state-of-the-art BnB MDCP algorithms calculate the upper bound quickly but conservatively as they ignore many possible missing edges. In this paper, we propose a novel CoLoring-based Upper Bound (CLUB) that uses graph coloring techniques to detect independent sets so as to detect missing edges ignored by the previous methods. We then develop a new BnB algorithm for MDCP, called KD-Club, using CLUB in both the preprocessing stage for graph reduction and the BnB searching process for branch pruning. Extensive experiments show that KD-Club significantly outperforms state-of-the-art BnB MDCP algorithms on the number of solved instances within the cut-off time, having much smaller search tree and shorter solving time on various benchmarks.
\end{abstract}

\section{Introduction}
\input{body/01-Intro}

\section{Preliminaries}
\label{sec-Pre}
\input{body/02-Pre}

%\section{CLUB: A Coloring-based Upper Bound}
\section{Coloring-based Upper Bound}
\input{body/03-UpperBound}

%\section{KDCL: A New BnB Algorithm}
\section{Branch and Bound Algorithm}
\input{body/04-BnB}

%\section{Experiments}
\section{Empirical Evaluation}
\input{body/05-Exp}

\section{Conclusion}
\input{body/06-Con}

\setcounter{secnumdepth}{0}

\section*{Acknowledgments}
This work is supported by National Natural Science Foundation (U22B2017).

\bibliography{aaai24}

\end{document}

%% file: body/01-Intro.tex
%The Maximum Clique Problem (MCP) is a classical NP-hard combinatorial optimization problem with applications such as bioinformatics~\cite{app-bio} and fault diagnosis~\cite{app-fault}. However, the high requirement of MCP on edge connectivity makes it unsuitable for many sparse real-world applications. In that case, some relaxations of the clique, such as $k$-plexes~\cite{k-plex-pro,k-plex-KPLEX,k-plex-kplexS}, quasi-cliques~\cite{quasi-cliques,quasi-cliques2,quasi-cliques3}, $k$-clubs~\cite{k-club,k-club2}, and $k$-defective cliques~\cite{k-defective-pro}, were proposed. Among these relaxations, $k$-defective clique is widely used in social network analysis~\cite{app-social} and transportation science~\cite{app-trans1,app-trans2}.

Investigating structured subgraphs is a practical task with numerous demands in many optimization problems and real-world applications. The clique model is a famous and well-studied structured subgraph where any two distinct vertices are restricted to be adjacent. However, in many real-world applications, such as biological networks~\cite{k-defective-pro}, 
social networks~\cite{k-plex-pro}, and community detection~\cite{ConteMSGMV18,Yang22KBS}, dense subgraphs need not be complete but allow missing a few connections. Thus, many relaxations clique structures, such as the quasi-clique~\cite{quasi-cliques-pro}, $k$-plex~\cite{k-plex-pro}, $k$-defective clique~\cite{k-defective-pro}, etc., have been proposed.

Among these relaxations, the $k$-defective clique allows missing at most $k$ edges over a clique, and the $k$-plex allows missing at most $k - 1$ edges for each vertex. Obviously, the relaxation of the $k$-defective clique is between the clique and the $k$-plex, \ie, a clique must be a $k$-defective clique, and a $k$-defective clique must be a $(k+1)$-plex. The clique problem has been well studied in the past decades. Recently, the $k$-plex also attracted much attention~\cite{k-plex-kplexS,WangZXK22,RGB,jiang2023Dise,wang2023fast}, yet there are relatively few studies on the $k$-defective clique~\cite{MADEC,KDBB,DaiLLW23} which also has wide applications, such as social network analysis~\cite{WWW_social_network}, transportation~\cite{app-trans2}, clustering~\cite{Clustering}, and protein interaction prediction~\cite{k-defective-pro}. In this paper, we address the Maximum $k$-Defective Clique Problem (MDCP), which aims to find the maximum $k$-defective clique in a given graph, and we focus on its exact solving.

%Let $G=(V,E)$ be an undirected graph, where $V$ is the set of vertices and $E$ is the set of edges. Given an integer $k$, a $k$-defective is a subgraph $S$ of $G$, allowing at most $k$ edges missing from a complete graph. The Maximum $k$-Defective Clique Problem (MDCP) is intended to find a $k$-defective with the largest number of vertices. The MCP is a special case of the MDCP at $k = 1$. Similar to the complexity of MCP, the MDCP is also a NP-hard combinatorial optimization problem.

Since the $k$-defective clique is a relaxation of the clique, the difficulty of solving MDCP is as hard as finding the maximum clique in a given graph, which is a famous NP-hard problem. 
Some exact algorithms for MDCP have been proposed, coming up with a series of efficient techniques, such as reduction rules and upper bounds. For representative methods, there are the generic RDS algorithm~\cite{RDS-proposed,RDS} and a branch-and-price framework~\cite{k-defective-bnb} for various relaxed clique problems, including MDCP, and the most recent MADEC$^+$~\cite{MADEC} and KDBB~\cite{KDBB} algorithms that proposed some new upper bounds and reduction rules. As the state-of-the-art exact algorithms for MDCP, both MADEC$^+$ and KDBB follow the branch-and-bound (BnB) framework~\cite{McCreeshPT17,JiangLLM18}.

%Recently, some new reduction and pruning strategies are proposed in the MADEC$^+$ algorithm~\cite{MADEC}, and the KDBB algorithm~\cite{KDBB} proposes some new upper bounds for MDCP. Both MADEC$^+$ and KDBB, state-of-the-art exact MDCP algorithms, follow the branch-and-bound (BnB) framework.

%The first exact algorithm~\cite{k-defectice-firstexact} was based on the Russian Doll Search (RDS)~\cite{RDS-proposed}. Then an improved RDS~\cite{RDS} was proposed by a new incremental verification procedure with a better worst-case runtime. In recent years, most of the exact algorithms follow the branch-and-bound (BnB) framework. In 2021, a BnB algorithm~\cite{k-defective-bnb} was proposed to solve MDCP and other relaxed cliques problems. Then some new reduction and pruning strategies were proposed in MADEC$^+$~\cite{MADEC}. In 2022, KDBB~\cite{KDBB} put forward some new upper bounds, which used in preprocessing stage, reduction stage and pruning stage.

A BnB algorithm for MDCP usually maintains a growing partial solution $S$ (\ie, a $k$-defective clique) and the corresponding candidate vertex set $C$. Reduction or pruning is performed if the upper bound of the size of the maximum $k$-defective clique containing $S$ is no larger than the size of the maximum $k$-defective clique found so far (\ie, lower bound). %Pruning is performed by calculating the upper bound of the current solution and comparing it with the lower bound currently obtained. 
%In that case, a high-quality upper bound has a great influence on the efficiency of the algorithm. 
Therefore, the quality of the upper bound greatly influences the algorithm's efficiency. MADEC$^+$ proposes an upper bound based on graph coloring. Note that an independent set in a graph is a vertex set where any two distinct vertices are non-adjacent. MADEC$^+$ uses graph coloring methods to assign each vertex color with the constraint that adjacent vertices cannot be in the same color, so as to partition the entire subgraph of $G$ induced by $V \backslash S$ into independent sets and then calculates the upper bound. Such an upper bound regards $S$ and $V \backslash S$ as entirely independent and increases significantly with the increase of $k$.

KDBB proposes some new upper bounds that focus on the missing edges between candidate vertices and vertices in $S$, which increase softly with the increase of $k$. Thus, KDBB shows excellent performance for MDCP instances with large values of $k$. The upper bounds in KDBB also lead to efficient reduction rules, helping KDBB significantly reduce massive sparse graphs. However, these upper bounds are still not very tight since they ignore the missing edges between candidate vertices. In other words, they regard the entire candidate set $C$ as a clique, which is over-conservative.
%due to ignoring the connection relationship between the point of the candidate set.

To address the above issues, we propose a new upper bound based on graph coloring, called \textbf{C}o\textbf{L}oring-based \textbf{U}pper \textbf{B}ound (\textbf{CLUB}), that considers the missing edges between not only vertices in $C$ and vertices in $S$ but also vertices in $C$ themselves. %The main ideas are as follows. 
The main idea is that for a subset $C'$ of the candidate set $C$, we use the graph coloring method to partition $C'$ into $r$ independent sets. Then, when we want to add $t$ vertices (suppose $r < t \leq |C'|$) from $C'$ to $S$, besides the missing edges between the $t$ vertices and vertices in $S$ that might exist, there must be missing edges between the $t$ added vertices since there must be added vertices belonging to the same independent set. Since CLUB considers the missing edges more thoroughly, it is strictly no worse than the upper bounds proposed in KDBB.
%the number of increased missing edges must be larger than the number of missing edges between the $t$ vertices and the vertices in $S$.
%is at least the number of missing edges between the $t$ vertices and the vertices in $S$ plus $t - ub$.

%To handle this issue, we propose a Unified Coloring-based Upper Bound (UCUB) for MDCP, which can be used for both the graph reduction in preprocessing and pruning branches during the search. For instance, we color the vertices in the candidate set, thinking there is a clique with the size of chromatic number $\chi$, instead of treating the candidate set as a clique. When there are more than $\chi$ vertices, an additional number of disconnected edges will be considered, resulting in a tighter upper bound.

Our proposed CLUB is a generic approach that can be used not only during the BnB searching process to prune the branches but also in preprocessing for graph reduction. Based on CLUB, we propose a new BnB algorithm for MDCP, called %\textbf{KDCL} 
\textbf{\name} 
(maximum \textbf{K}-\textbf{D}efective clique algorithm with \textbf{CLUB}). Since the upper bound in MADEC$^+$ increases significantly with the increase of $k$, MADEC$^+$ does not work well for MDCP instances with large $k$ values. Since the reduction rules in KDBB cannot help it reduce dense graphs, KDBB does not work well for MDCP instances based on dense graphs. Our proposed CLUB significantly outperforms the previous upper bounds for MDCP by considering the connectivity between vertices more thoroughly, helping the BnB algorithm significantly reduce the search space. Therefore, \name has excellent performance and robustness for solving MDCP instances based on either massive sparse or dense graphs, with either small or large $k$ values, as shown in our experiments.

%% file: body/02-Pre.tex
%We only consider simple and undirected graphs in this paper. 
%This section first provides some definitions and notations, and then introduces some reduction rules proposed in the KDBB algorithm~\cite{KDBB}, which will also be used in the preprocessing of our \name algorithm. 
%This section introduces some definitions, and two reduction rules used in the preprocessing of the KDBB algorithm.%~\cite{KDBB}.

\subsection{Definitions}
Given an undirected graph $G=(V,E)$, where $V$ is the vertex set and $E$ the edge set, the density of $G$ is $2|E|/(|V|(|V|-1))$. %Each edge $e \in E$ is denoted by its two endpoints $(u,v)$. 
We define $\overline{G} = (V,\overline{E})$ as the complementary graph of $G$, \ie, $\overline{E} = \{(u,v) | u \in V \wedge v \in V \wedge (u,v) \notin E\}$. 
%The vertex set and edge set of graph $G$ can also be denoted as $V(G)$ and $E(G)$, respectively. %$(u,v)\in \overline{E}$ iff $(u,v) \notin E$. %The density of $G$ is computed as $2m/(n(n-1))$.

We refer to adjacent vertices as neighbors to each other and denote the set of neighbors to $v$ in $G$ as $N_G(v)$, $|N_G(v)|$ is the degree of $v$ in $G$, 
and the common neighbor of two vertices $u, v$ as $N_G(u,v) = N_G(u)\cap N_G(v)$. Moreover, we define $N_G[v] = N_G(v)\cup \{v\}$ and $N_G[u,v] = N_G(u,v)\cup \{u,v\}$. %, and $\oplus N_G(u,v) = N_G(u)\cup N_G(v) - N_G[u,v]$ as the set of vertices adjacent to exactly one of the two vertices $u$ and $v$. 
Given a vertex set $V' \subseteq V$, $G[V']$ is defined as the subgraph induced by $V'$. Given a positive integer $k$, $V' \subseteq V$ is a $k$-defective clique if $G[V']$ has at least ${|V'|\choose 2}-k$ edges.

%Moreover, we define $N_G[v] = N_G(v)\cup \{v\}$. Similarly, we denote the common neighbor of two vertices $u$ and $v$ as $N_G(u,v) = N_G(u)\cap N_G(v)$, and $N_G[u,v] = N_G(u,v)\cup \{u,v\}$. Besides, $\oplus N_G(u,v) = N_G(u)\cup N_G(v) - N_G[u,v]$ denotes the set of vertices connected to only one of the two vertices $u$ and $v$. 

%In the rest of this paper, 
We use $S\subseteq V$ to denote a growing partial solution (\ie, a growing $k$-defective clique) in BnB algorithms, and $C \subseteq V\backslash S$ as the corresponding candidate vertex set of $S$. The size of the maximum $k$-defective clique in $G$ that includes all vertices in $S$ is denoted by $\omega_{G,k}(S)$, and the size of the maximum $k$-defective clique in $G$ is $\omega_{G,k}(\emptyset)$.

%In the BnB algorithm, $S\subseteq V$ is the growing partial solution and $d_S(v)$ is the number of neighbors of $v$ in $S$. $G[S] = (S,E')$ is the induced subgraph in $G$ by $S$, where $E'$ contains edges whose both vertices are in $S$. Besides, we define the candidate set as $C$, where $C\subseteq V\backslash S$ and we define $\omega_{G,k}(S)$ as the size of the maximum $k$-plex in $G$ that includes all vertices in $S$. Give a positive integer $k$, $G[S]$ is a $k$-defective if $G[S]$ has at least ${|S|\choose 2}-k$ edges, the size of the maximum $k$-defective in $G$ denoted by $\kappa(G)$ in this paper.

\subsection{Revisiting the Reduction Rules of KDBB}

The KDBB algorithm defines function $rem_{G,k}(S,v) = k-|E(\overline{G}[S\cup\{v\}])|$ as the remaining number of allowed missing edges of $S \cup \{v\}$ to form a feasible $k$-defective clique, and function $res_{G,k}(S,v) = \min\{rem_{G,k}(S,v), |C\backslash N_{G[C\cup \{v\}]}[v]|\}$ as the maximum number of extra vertices in $C$ non-adjacent to $v$ that are allowed to be added to $S \cup \{v\}$. KDBB proposes an upper bound of $\omega_{G,k}(S \cup \{v\})$ calculated as $UB_{G,k}(S,v) = |S\cup\{v\}|+|N_{G[C]}(v)|+res_{G,k}(S,v)$, which actually considers $C$ as a clique and adding a vertex from $C\backslash N_{G[C\cup \{v\}]}[v]$ to $S$ only leads to one more missing edge.

%regards $N_{G[C]}(v)$ as a clique. 
%follows.
%\begin{equation}
%\label{UB1}
%    UB_{G,k}(S,v) = |S\cup\{v\}|+|N_{G[C]}(v)|+res_{G,k}(S,v).
%\end{equation}

Similarly, KDBB defines function $rem_{G,k}(S,u,v) = k-|E(\overline{G}[S\cup\{u,v\}])|$ as the remaining number of allowed missing edges of $S \cup \{u,v\}$ to form a feasible $k$-defective clique, and function $res_{G,k}(S,u,v) = \min\{rem_{G,k}(S,u,v), |C\backslash N_{G[C\cup \{u,v\}]}[u,v]|\}$ as the maximum number of extra vertices in $C$ non-adjacent to $u$ or $v$ that are allowed to be added to $S \cup \{u,v\}$. And KDBB proposes another upper bound of $\omega_{G,k}(S \cup \{u,v\})$ calculated as $UB_{G,k}(S,u,v) = |S\cup\{u,v\}|+|N_{G[C]}(u,v)|+res_{G,k}(S,u,v)$, regarding $C$ as a clique also.
%follows.
%\begin{equation}
%\label{UB2}
%\begin{aligned}
%    UB_{G,k}&(S,u,v) = \\&|S\cup\{u,v\}|+|N_{G[C]}(u,v)|+res_{G,k}(S,u,v).
%\end{aligned}
%\end{equation}

Given a lower bound $LB$ of $\omega_{G,k}(\emptyset)$, KDBB proposes the following two reduction rules based on the above bounds.

\textbf{Rule 1.} Remove vertex $v$ from $G$ if it satisfies $UB_{G,k}(\emptyset,v)\le LB$.

\textbf{Rule 2.} Remove edge $(u,v)$ from $G$ if it satisfies $UB_{G,k}(\emptyset,u,v)\le LB$.

%% file: body/03-UpperBound.tex
This section introduces our proposed upper bound, CLUB. We first introduce the main idea and definition of CLUB, then present an example for illustration. In the end, we introduce our ColorBound algorithm to calculate the CLUB.

%Given a graph $G=(V,E)$, and an integer $k$, let $S\subseteq V$ be a partial solution for MDCP in $G$ and $C$ be the candidate set. $\gamma$ denotes the current number of missing edges in $S$.

\subsection{Main Idea}%s of CLUB}
%We define $\{P_0,\cdots,P_k\}$ as a partition of $C$, and for each vertex in $P_i$, if it is added to $S$, at least $i$ missing edges will be added. We perform the following steps for $C$ to obtain such a division.

Suppose $\{C_0,\cdots,C_k\}$ is a partition of $C$ \textit{w.r.t.} $S$, where $C_i$ is the set of vertices with $i$ non-adjacent vertices in $S$, \ie, $C_i = \{v | v \in C \wedge |N_G(v) \cap S| = |S| - i\}$. In other words, adding each vertex $v \in C_i$ to $S$ leads to $i$ missing edges between $v$ and vertices in $S$. Then, we can summarize the following important lemma which implies our main idea.

\begin{lemma}
\label{lemma-1}
    Suppose $C_i$ can be partitioned into $r_i$ independent sets $\{I_{i,1},\cdots,I_{i,r_i}\}$, adding any $1 \leq t \leq |C_i|$ vertices in $C_i$ to $S$ leads to at least $\mathcal{N}(t) = c\times \frac{d(d+1)}{2} + (r_i - c) \times \frac{d(d-1)}{2}$ more missing edges between the $t$ added vertices, where $d = \lfloor t / r_i \rfloor$ and $c = t - r_i d$.
\end{lemma}

\begin{proof}
    Suppose the $t$ vertices contain $d_j$ vertices in $I_{i,j}$, then the number of missing edges between the $t$ vertices is at least $\sum_{j=1}^{r_i}\frac{d_j(d_j-1)}{2} = \frac{1}{2}\left(\sum_{j=1}^{r_i}d_j^2 - t\right)$. Obviously, to make the above lower bound as small as possible, every $d_j$ is expected to be $t/r_i$. Since $d_j$ must be an integer, every $d_j$ should be either $\lceil t / r_i \rceil$ or $\lfloor t / r_i \rfloor$ to minimize the lower bound. In this case, there are $t \% r_i$ independent sets containing $\lceil t / r_i \rceil$ vertices in $t$ and the rest of the $r_i - t \% r_i$ independent sets contain $\lfloor t / r_i \rfloor$ vertices in $t$, which results in the lower bound in Lemma~\ref{lemma-1}.
\end{proof}

%One can see that 
Lemma~\ref{lemma-1} provides a lower bound, \ie, $\mathcal{N}(t)$, of the number of increased missing edges, with which we can calculate an upper bound of the number of vertices that $C$ can provide for $S$, resulting in the proposed CLUB. 

Lemma~\ref{lemma-1} also indicates that the smaller the value of $r_i$, the more missing edges by adding vertices in $C_i$ to $S$. KDBB actually regards $r_i = |C_i|$ since it regards $C$ as a clique for the bound calculation, while our CLUB uses graph coloring methods to partition $C_i$ and determine the value of $r_i$. So CLUB is strictly no worse than the upper bounds in KDBB.

%Firstly, adding $t \leq r_i$ vertices in $C_i$ to $S$ might lead to no missing edges between the $t$ vertices, since each vertex might belong to an independent set and the set of the $t$ vertices might be a clique. Secondly, adding $r_i + 1$ vertices in $C_i$ to $S$ leads to at least one missing edge between the $r_i + 1$ vertices, since there must be two vertices belonging to the same independent set. More generally, we can summarize the following important Lemma, which implies our main ideas.

%Since after adding $t$ vertices in $C_i$ to $S$, any other vertex $v \in C_i$ must be 

%suppose $t > r_i$ vertices in $C_i$ have been added to $S$, adding any other vertex $v \in C_i$ to $S$ leads to at least $\lfloor t / r_i \rfloor$ missing edges between $v$ and the $t$ vertices. Obviously, the smaller the value of $r_i$, adding vertices in $C_i$ to $S$ leads to more missing edges. The KDBB algorithm actually regards $r_i = |C_i|$, and our CLUB uses graph coloring methods to partition $C_i$ and determine the value of $r_i$.

\subsection{Definition of CLUB}

For each set $C_i$ that can be partitioned into $r_i$ independent sets, we randomly re-partition $C_i$ into $m_i = \lceil |C_i| / r_i \rceil$ disjoint subsets $\{I'_{i,1},\cdots,I'_{i,m_i}\}$. If $r_i$ divides $|C_i|$, each of the $m_i$ sets contains $r_i$ vertices in $C_i$. Otherwise, $I'_{i,m_i}$ contains $|C_i| \% r_i$ vertices in $C_i$ and each of the remaining $m_i - 1$ sets contains $r_i$ vertices in $C_i$. We assume that vertices in $I'_{i,j}$ $(1 < j \leq m_i)$ can be added to $S$ only when all vertices in $\cup_{l=1}^{j-1}I'_{i,l}$ have been added to $S$.
%, and try to use Lemma 1 for counting the least number of missing edges caused by adding vertices from $C_i$ to $S$.

According to Lemma 1, when we add $r_i$ vertices in $I'_{i,1}$ to $S$, we have $\mathcal{N}(1) = \cdots = \mathcal{N}(r_i) = 0$, which may only occur when the $r_i$ vertices belong to different independent sets. When we further add $r_i$ vertices in $I'_{i,2}$ to $S$, we have $\mathcal{N}(r_i + 1) - \mathcal{N}(r_i) = \mathcal{N}(r_i + 2) - \mathcal{N}(r_i + 1) = \cdots = \mathcal{N}(2r_i) - \mathcal{N}(2r_i - 1) = 1$. Similarly, we have $\mathcal{N}(jr_i + 1) - \mathcal{N}(jr_i) = \mathcal{N}(jr_i + 2) - \mathcal{N}(jr_i + 1) = \cdots = \mathcal{N}((j+1)r_i) - \mathcal{N}((j+1)r_i - 1) = j$. In summary, adding any vertex in $I'_{i,j+1}$ to $S$ leads to at least $i + j$ more missing edges.

%We assume that each subset $I'_{i,j}$ is a clique, which only occurs when vertices in $I'_{i,j}$ belong to different independent sets. We further assume that vertices in $I'_{i,j}$ $(1 < j \leq m_i)$ can be added to $S$ only when all vertices in $\{I'_{i,1},\cdots,I'_{i,j-1}\}$ have been added to $S$. Actually, these assumptions lead to the lower bound introduced in Lemma~\ref{lemma-1}. In this case, adding each vertex $v \in I'_{i,j}$ to $S$ leads to $j-1$ missing edges between $v$ and vertices in $\{I'_{i,1},\cdots,I'_{i,j-1}\}$, and $j-1+i$ missing edges between $v$ and vertices in $S$.

Suppose $P_l = \cup_{j-1+i = l}{I'_{i,j}}$, then adding each vertex $v \in P_l$ to $S$ leads to $l$ more missing edges. We define function $Sub(v) = l | v \in P_l$ as the subscript of the subset that $v \in C$ belongs to, and define an ordered set of $C$ as $ord(C) = \{v_1,\cdots,v_{|C|}\}$ such that for any pair of $v_i$ and $v_j$, we have $Sub(v_i) \leq Sub(v_j) $ if $i < j$. With such an ordered set, we can calculate a lower bound on the number of increased missing edges caused by adding any $t$ vertices of $C$ to $S$, which is defined as $LB_{inc}(ord(C),t) = \sum_{i=1}^t{Sub(v_i)}$.

%We then extract $k+1$ disjoint subsets $\{P_0,\cdots,P_k\}$ from the candidate set $C$, where $P_l = \{v | v \in I'_{i,j} \wedge j-1+i = l\} = \cup_{j-1+i = l}{I'_{i,j}}$, \ie, adding each vertex $v \in P_l$ to $S$ leads to at least $l$ more missing edges. We denote $\mathcal{P} = \cup_{l=0}^k{P_l}$. Actually, adding any vertex $v \in C \backslash \mathcal{P}$ to $S$ leads to more than $k$ missing edges. Thus the candidate set can be reduced to be $\mathcal{P}$. We further define function $Part(v) = l | v \in P_l$ as the subscript of the subset that $v \in \mathcal{P}$ belongs to, and define an ordered set of $\mathcal{P}$ as $ord(\mathcal{P}) = \{v_1,\cdots,v_{|\mathcal{P}|}\}$ such that for any pair of $v_i$ and $v_j$, we have $Part(v_i) \leq Part(v_j) $ if $i < j$. With such an ordered set, we can calculate a lower bound on the number of increased missing edges caused by adding any $t$ vertices of $\mathcal{P}$ to $S$, which is defined as $LB_{inc}(ord(\mathcal{P}),t) = \sum_{i=1}^t{Part(v_i)}$.

Finally, we define our CLUB of $\omega_{G,k}(S)$ as follows.
\iffalse
\begin{equation}
\label{CLUB}
CLUB_{G,k}(S) = |S| + \mathcal{I},
\end{equation}
where 
\begin{equation}
\label{CLUB2}
\mathcal{I} = \max_{0 \leq i \leq |C_i|}{LB_{inc}(ord(C),i) \leq k - |E(\overline{G}[S])|}
\end{equation}
is an upper bound on the number of vertices that $C$ can provide for $S$ to form a feasible $k$-defective clique.
\fi

\begin{equation}
\label{CLUB}
\begin{aligned}
    %CLUB(G,S,k) = |S| + \\\max_{0 \leq i \leq |C_i|}{LB_{inc}(ord(C),i) \leq k - |E(\overline{G}[S])|}.
    CLUB_{G,k}(S) = |S| &+ \max\{i|0 \leq i \leq |C| \\&\wedge LB_{inc}(ord(C),i) \leq k - |E(\overline{G}[S])|\}.
    %0 \leq i \leq |C_i|}{LB_{inc}(ord(C),i) \leq k - |E(\overline{G}[S])|}.
\end{aligned}
\end{equation}

\subsection{An Example for Illustration}
In this subsection, we provide an example to show how the upper bound in KDBB and CLUB are calculated. Figure~\ref{fig:example} illustrates a growing partial 1-defective clique $S = \{v_0\}$ and its candidate set $C = \{v_1,v_2,v_3,v_4,v_5,v_6\}$, which can be partitioned into $C_0 = \{v_1,v_2,v_3,v_4,v_5\}$ and $C_1 = \{v_6\}$. Suppose the current lower bound $LB$ of $\omega_{G,1}(\emptyset)$ is 6.

KDBB regards $C_0$ and $C_1$ as cliques, \ie, adding all vertices in $C_0$ to $S$ does not increase any missing edges, and adding each vertex in $C_1$ to $S$ leads to one more missing edge. Thus, KDBB considers the upper bound of $\omega_{G,1}(S) = |S| + 6 = 7 > LB$, and the current branch cannot be pruned, \ie, $v_0$ will not be removed from the graph. 
%Calculating upper bound according to the method of KDBB, the number of neighbors of $v_0$ is 5, and one non-neighbor of $v_0$ can be added. Thus, the upper bound is 6, which is larger than the current lower bound 4.

In CLUB, $C_0$ and $C_1$ are partitioned into 3 and 1 independent sets, respectively (\eg, $I_{0,1} = \{v_1,v_4\}, I_{0,2} = \{v_2\}, I_{0,3} = \{v_3,v_5\}$, and $I_{1,1} = \{v_6\}$). Thus, we can re-partition $C_0$ into $I'_{0,1} = \{v_1,v_2,v_3\}, I'_{0,2} = \{v_4,v_5\}$, and $C_1$ into $I'_{1,1} = \{v_6\}$. And we have $P_0 = I'_{0,1} = \{v_1,v_2,v_3\}, P_1 = I'_{0,2} \cup I'_{1,1} = \{v_4,v_5,v_6\}$. According to Eq.~\ref{CLUB}, $CLUB_{G,1}(S) = |S| + 4 = 5 < LB$, and the current branch can be pruned, \ie, $v_0$ will be removed. % from the graph. 

%After coloring $C_0$, the color number is 3. In other words, only three vertices of $C_0$ do not donate missing edges, and other two vertices each donate one missing edge. Above all, the CLUB is 4, which is smaller than the current lower bound. The vertex $v_0$ can be removed or the branch can be cut.
\begin{figure}[t]
    \centering
    \includegraphics[width=0.9\columnwidth]%[width=1.0\columnwidth]
    {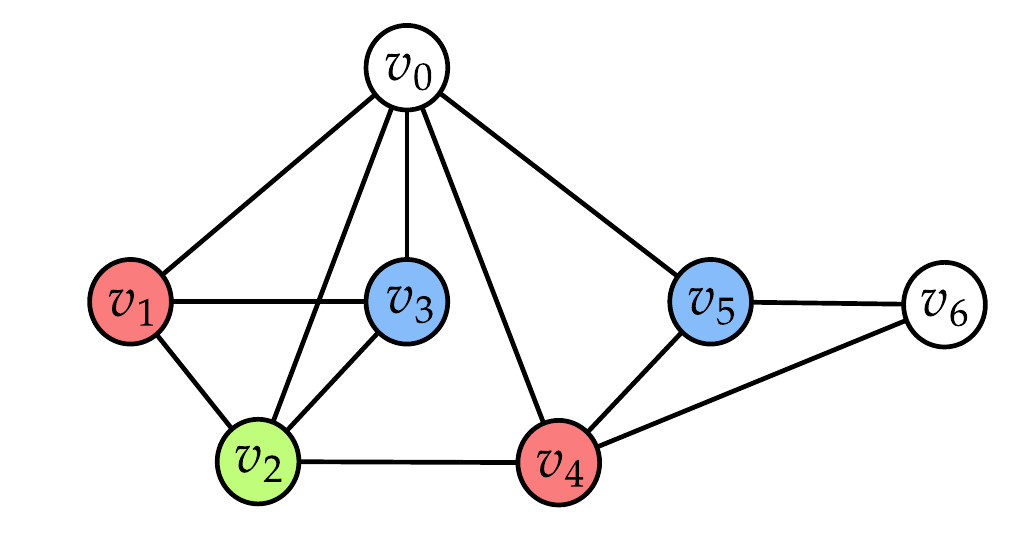}
    \caption{An example for %comparing
     the upper bound calculation.}
    \label{fig:example}
\end{figure}

\subsection{The ColorBound Algorithm}
We propose the ColorBound algorithm for practically calculating the proposed CLUB, as summarized in Algorithm~\ref{alg-ColorBound}. %The algorithm receives the MDCP instance $G=(V,E)$, an integer $k$, the current partial $k$-defective, the candidate set $C$, and the current number of missing edges $\gamma$ as inputs, and outputs the CLUB of $\omega_k(G,S)$. 
The algorithm first uses $|S|$ to initialize the upper bound $UB$ (line 1), then uses the Extract function to extract subsets $\{P_0,\cdots,P_k\}$ from the candidate set $C$ (line 3), such that adding each vertex in $P_l$ to $S$ leads to at least $l$ more missing edges when adding the vertices sequentially according to $ord(C)$. Note that subsets $P_l$ with $l > k$ are ignored since there are at most $k$ missing edges in a $k$-defective clique. After that, CLUB is calculated by simulating the process of adding vertices sequentially until the number of allowed missing edges cannot afford one more vertex or all the candidate vertices have been added (lines 4-9).

Function Extract is summarized in Algorithm~\ref{alg-ColorBasedPartition}, which first partitions $C$ into subsets $\{C_0,\cdots,C_k\}$ according to the number of missing edges between each vertex and the vertices in $S$ and initializes each set $P_i$ to $\emptyset$ (lines 1-3). Then, for each subset $C_i \neq \emptyset$, the algorithm sequentially colors each vertex in $C_i$ with the minimum feasible index of color, satisfying that adjacent vertices cannot be in the same color, and obtains the value of $r_i$ (lines 4-6). Finally, the algorithm iteratively moves $r_i$ vertices, \ie, set $I'_{i,j}$, from $C_i$ to $P_{j-1+i}$ until $C_i$ is empty or $j-1+i > k$ (lines 8-12), which means adding any vertex in $I'_{i,j}$ to $S$ according to $ord(C)$ leads to more than $k$ missing edges.

%Finally, the algorithm uses the color number to move the vertices in the subset to the corresponding partition sets.

The time complexities of both Algorithms \ref{alg-ColorBound} and \ref{alg-ColorBasedPartition} are dominated by the graph coloring process, which needs to traverse each candidate vertex and its neighbors. Thus, their time complexities are both $O(D|C|^2)$, where $D$ is the maximum degree of vertices in $G$.

%The time complexities of both the ColorBound algorithm and the Partition() function are $O(d|C|)$, where $d$ is the maximum degree of $G$.

\begin{algorithm}[t]
\caption{ColorBound$(G,k,S,C)$}
\label{alg-ColorBound}
\LinesNumbered 
\KwIn{a graph $G$, an integer $k$, the current $k$-defective clique $S$, the candidate set $C$}
\KwOut{$CLUB_{G,k}(S)$}
initialize the upper bound $UB \leftarrow |S|$\;
initialize the number of allowed missing edges $\kappa \leftarrow k - |E(\overline{G}[S])|$\;
%$P\leftarrow $Partition($G,k,S,C$)\;
$\{P_0,\cdots,P_k\} \leftarrow $Extract($G,k,S,C$)\;
$UB \leftarrow UB + |P_0|$\;
\For{$i \leftarrow 1 : k$}{
\eIf{$i\times |P_i| \le \kappa$}{
$UB\leftarrow UB+|P_i|$, $\kappa \leftarrow \kappa - i\times|P_i|$\;
}{$UB\leftarrow UB + \lfloor \kappa/i \rfloor$, \textbf{break}\;}
}
\textbf{return} $UB$\;
\end{algorithm}

\begin{algorithm}[t]
\caption{Extract$(G,k,S,C)$}
\label{alg-ColorBasedPartition}
\LinesNumbered 
\KwIn{a graph $G$, an integer $k$, the current $k$-defective clique $S$, the candidate set $C$}
\KwOut{subsets $\{P_0,\cdots,P_k\}$ of $C$}
%$P\leftarrow \{P_0,P_1,\cdots,P_k\}, P_i = \emptyset$\;
%$C'\leftarrow \{C_0,C_1,\cdots,C_k\}, C_i = \emptyset$\;
%\For{\rm{\textbf{each}} vertex $v \in C$} {$g \leftarrow k-d_S(v)$\;
%$C_g\leftarrow C_g\cup\{v\}$}
\For{$i \leftarrow 0 : k$}{
$C_i \leftarrow \{v | v \in C \wedge |N_G(v) \cap S| = |S| - i\}$\;
initialize $P_i \leftarrow \emptyset$\;
}
\For{$i \leftarrow 0 : k$}{
\If{$C_i\neq \emptyset$}{
partition $C_i$ into $r_i$ independent sets by sequentially coloring the vertices\;
%$r\leftarrow $ graph\_coloring($G,C_i$)\;
$j\leftarrow 1$\;
\While{$|C_i|\geq r_i \wedge j - 1 + i \leq k$}{
$I'_{i,j} \leftarrow$ set of $r_i$ vertices in $C_i$\;
$C_i \leftarrow C_i \backslash I'_{i,j}$, $P_{j - 1 + i} \leftarrow P_{j - 1 + i} \cup I'_{i,j}$\;
%Move $r_i$ vertices from $C_i$ to $P_j$\;
$j\leftarrow j+1$\;
}
\lIf{$j - 1 + i \leq k$}{$P_{j - 1 + i} \leftarrow P_{j - 1 + i}\cup C_i$}
}
}
\textbf{return} $\{P_0,\cdots,P_k\}$\;
\end{algorithm}

%% file: body/04-BnB.tex
We propose a new BnB algorithm for MDCP, called \name, where a new preprocessing method based on CLUB is introduced to reduce the graph, and the CLUB is also used to prune branches during the BnB process. We first present the main framework of \name, and then introduce the preprocessing method and the BnB process, respectively.

\subsection{General Framework}
The framework of \name is summarized in Algorithm~\ref{alg-Framework}. \name maintains a lower bound $LB$ of the size of the maximum $k$-defective clique in the input graph $G$, which is initialized by a method called \textit{FastLB} (line 1), which is also used in the MADEC$^+$~\cite{MADEC} and KDBB~\cite{KDBB} algorithms for calculating an initial lower bound. After calculating $LB$, the Preprocessing() function is called to reduce the graph (line 2), and the reduced graph is sent to the BnB() function to find the maximum $k$-defective clique (line 3). During the BnB process, $LB$ will be updated once a larger $k$-defective clique is found. After the BnB() function traverses the entire search tree, we have $LB = \omega_{G,k}(\emptyset)$.
%To find an initial lower bound, we employ the method called $FastLB$, which is often used to find a clique as the lower bound of MDCP~\cite{MADEC, KDBB}. Then the function $Preprocessing$ is called to reduce the graph. After that, the reduced graph is fed into the $BnB$ algorithm to obtain the maximum $k$-defective. In the $BnB$ function, it checks whether $S$ can form a larger $k$-defective than the current lower bound, and a binary branching strategy is used to search the maximum $k$-defective.

\begin{algorithm}[t]
\caption{\name$(G,k)$}
\label{alg-Framework}
\LinesNumbered 
\KwIn{a graph $G$, an integer $k$}
\KwOut{$\omega_{G,k}(\emptyset)$}
%$S\leftarrow \emptyset$\;
initialize $LB$ by a heuristic method \textit{FastLB}\;
$G\leftarrow$ Preprocessing$(G,k)$\;
$LB\leftarrow$ BnB$(G,k,\emptyset,null,LB)$\; 
\textbf{return} $LB$\;
\end{algorithm}

\begin{algorithm}[t]
\caption{Preprocessing$(G,k)$}
\label{alg-Preprocessing}
\LinesNumbered 
\KwIn{a graph $G$, an integer $k$}
\KwOut{the reduced graph $G$}
%$removed\leftarrow true$\;
$G\leftarrow check\_vertex\_with\_Rule1(G,k,LB)$ \;
$G\leftarrow check\_vertex\_with\_Rule3(G,k,LB)$ \;
$G\leftarrow check\_edge\_with\_Rule2(G,k,LB)$ \;
\While{true} {
%$removed\leftarrow false$\;
$G'\leftarrow check\_vertex\_with\_Rule3(G,k,LB)$ \;
$G'\leftarrow check\_edge\_with\_Rule4(G',k,LB)$ \;
\lIf{$G'$ and $G$ is the same}{\textbf{break}}
\textbf{else} $G \leftarrow G'$\;
}
\textbf{return} $G$\;
\vspace{0.5em}
Function $check\_vertex\_with\_Rule3(G,k,LB)$\\
$Q\leftarrow\emptyset$, add all vertices in $V$ to $Q$\;
\While{$Q$ is not empty} {
$v\leftarrow pop(Q)$, $N_v\leftarrow N_G(v)$\;
apply Rules 1 and 3 on $v$ with $LB$ and $k$\;
%$ub\leftarrow CLUB_{G,k}(\{v\})$\;
\If{$v$ is removed}{add all vertices in $N_v$ to $Q$\;}
%$ub\leftarrow ColorBound(G,k,\{v\},V\backslash\{v\})$\;
%\If{$ub\le LB$}{
%remove $v$ from $G$, add all vertices in $N_v$ to $Q$\;
%}
%apply Rule 1 on $v$ with $LB$ and $k$\;
%\If{$v$ is not removed}{
%$ub\leftarrow ColorBound(G,k,\{v\},V\backslash\{v\},0)$\;
%\If{$ub\le LB$}{remove $v$ from $V$\;}
%}
%\If{$v$ is removed}{add all vertices in $N_v$ to $Q$\;}
}
\textbf{return} $G$\;
\vspace{0.5em}
Function $check\_edge\_with\_Rule4(G,k,LB)$\\
$Q\leftarrow\emptyset$;\ add all edges in $E$ to $Q$\;
\While{$Q$ is not empty} {
$(u,v)\leftarrow pop(Q)$\;
apply Rules 2 and 4 on $(u,v)$ with $LB$ and $k$\;
\If{$(u,v)$ is removed}{add all edges adjacent to $u$ or $v$ in $E$ to $Q$\;}
%apply Rule 2 on $(u,v)$ with $LB$ and $k$\;
%\If{$(u,v)$ is not removed}{
%$ub\leftarrow ColorBound(G,k,\{u,v\},V\backslash\{u,v\},0)$\;
%\If{$ub\le LB$}{remove $(u,v)$ from $E$\;}
%}
%\If{$(u,v)$ is removed}{add all edges of $u$ and $v$ to $Q$\;}
}
\textbf{return} $G$\;
\end{algorithm}

\subsection{Preprocessing Method}
Preprocessing plays an essential role in solving massive sparse instances. Given a lower bound $LB$ of $\omega_{G,k}(\emptyset)$, we propose two new rules to reduce the graph based on CLUB.

\textbf{Rule 3.} Remove vertex $v$ from $G$ if it satisfies $CLUB_{G,k}(\{v\})\le LB$.

\textbf{Rule 4.} Remove  edge $(u,v)$ from $G$ if it satisfies $CLUB_{G,k}(\{u,v\})\le LB$.

Algorithm~\ref{alg-Preprocessing} shows the detailed procedure of the Preprocessing method, where four functions with Rules 1-4 are used to reduce the input graph $G$. The functions $check\_vertex\_with\_Rule1$ and $check\_edge\_with\_Rule2$ are derived from KDBB, which uses Rules 1 and 2 to remove vertices and edges from $G$, respectively. Similarly, the functions $check\_vertex\_with\_Rule3$ and $check\_edge\_with\_Rule4$ use our proposed Rules 3 and 4 to remove vertices and edges from $G$.

Since Rules 1 and 2 ignore the connectivity between vertices in the candidate set $C$ (we have $C = V \backslash \{v\}$ when trying to remove vertex $v$), they are computationally efficient for removing vertices with small degrees or edges whose endpoints have small degrees. Hence, our Preprocessing method uses Rules 1 and 2 to quickly remove vertices and edges that are \textit{easy} to remove before using our Rules 3 and 4 to further reduce the graph until no more vertices and edges can be removed. Similarly, before calculating our CLUB to reduce the graph, we also try to use Rules 1 and 2 to achieve the current reduction and save computation time (lines 14 and 22). In other words, CLUB is used only when the vertex or edge cannot be removed by Rules 1 and 2.

%The functions $check\_vertex$ and $check\_edge$ are proposed in KDBB~\cite{KDBB}, and we only use these two function to obtain a smaller graph, so that we can further reduce the graph using CLUB. Besides, we alternately remove vertices and edges until no redundant vertices and edges exist.

%Two sub-functions based on CLUB are introduced here for vertices and edges reduction. However, if CLUB is used when removing every vertex, the time complexity is $O(d|V|^2)$. To reduce the graph more efficiently, CLUB is used only if the rules in KDBB can't remove the vertex or the edge. With the initial bound, Rule 1 and the function $ColorBound$ are performed iteratively in function $check\_vertex\_with\_CLUB$ to reduce the input graph, and Rule 2 and the function $ColorBound$ are performed iteratively in function $check\_edge\_with\_CLUB$ to remove redundant edges.

\subsection{Branch and Bound Process}
The BnB process is depicted in Algorithm~\ref{alg-BnB}. The algorithm first calls function $reduction()$ used in KDBB~\cite{KDBB} to reduce the input graph $G$ according to the current solution $S$ (line 1), which actually removes each vertex $v \in V \backslash S$ (resp. edge $(u,v) \in E$) if $S \cup \{v\}$ (resp. $S \cup \{u,v\}$) is not a feasible solution. After which, we have the candidate set $C = V \backslash S$ (line 3). Then, the algorithm calculates CLUB and checks whether it is larger than the current lower bound (lines 4-5). If so, the algorithm will continue to search the subtree. Otherwise, the branch will be pruned.

For the selection of the branching vertex, we select the vertex in $C$ with the minimum degree (ties are broken randomly) to find larger $LB$ quickly and reduce the tree size for the subsequent searching. After selecting a branching vertex $u$, the algorithm uses a binary branching strategy, that is, either adding $u$ to $S$ or removing it from $G$ (lines 7-11).

\begin{algorithm}[t]
\caption{BnB$(G,k,S,v,LB)$}
\label{alg-BnB}
\LinesNumbered 
\KwIn{a graph $G$, an integer $k$, the current $k$-defective clique $S$, the branching vertex $v$, the lower bound $LB$}
\KwOut{$\omega_{G,k}(S)$}
%\lIf{$|E(\overline{G}[S])|>k$}{\textbf{return} $LB$}
\lIf{$v\neq null$}{$G\leftarrow reduction(G,k,S,v,LB)$}
%\lIf{$V\backslash S=\emptyset$}{\textbf{return} $|S|$}
\lIf{$|V| \leq LB$}{\textbf{return} $LB$}
$C \leftarrow V\backslash S$\;
$ub\leftarrow$ ColorBound$(G,k,S,C)$\;
\If{$ub>LB$}{
select a vertex $u$ in $C$ with the minimum degree\; %and update $\gamma$\;
$size\leftarrow$ BnB$(G,k,S\cup\{u\},u,LB)$\;
\lIf{$size > LB$}{$LB\leftarrow size$}
remove $u$ from $G$\; %, restore $G$\; %and update $\gamma$\;
$size\leftarrow$ BnB$(G,k,S,u,LB)$\;
\lIf{$size>LB$}{$LB\leftarrow size$}
}
\textbf{return} $LB$\;
\end{algorithm}

%% file: body/05-Exp.tex
In this section, we first introduce the benchmarks and algorithms (also called solvers) used in the experiments, then present and analyze the experimental results. All the algorithms were implemented in C++ and run on a server using an AMD EPYC 7H12 CPU, running Ubuntu 18.04 Linux operation system. Since our machine is about 5-10 times faster than the machine in~\cite{KDBB}, which sets the cut-off time to 10,800 seconds for each instance, we set the cut-off time to 1,800 seconds in our experiments. % We select the state-of-art BnB algorithms, including MADEC$^+$\footnote{https://github.com/chenxiaoyu233/k-defective}~\cite{MADEC$^+$} and KDBB~\cite{KDBB}. Especially, the MADEC$^+$ always fails to solve the instances when $k>3$ without the preprocessing methon. Thus, we combine MADEC$^+$ with our preprocessing method to make an intensive analysis. Since KDBB does not disclose the source codes, we use the code that we reproduced according to the paper, whose results are much better than those broadcast in the paper.

\subsection{Benchmark Datasets}
We evaluated the algorithms on four public datasets that are widely used in existing works. 
\begin{itemize}
\item \textbf{Facebook\footnote{https://networkrepository.com/socfb.php}}: This dataset contains 114 massive sparse graphs derived from Facebook social networks, which is used in KDBB~\cite{KDBB}.
\item \textbf{Realword\footnote{http://lcs.ios.ac.cn/\%7Ecaisw/Resource/realworld\%20\\graphs.tar.gz}}: This dataset contains 139 massive sparse graphs from the Network Data Repository~\cite{RA15}, which is frequently used in studies related to relaxation clique models, including the $k$-defective clique and $k$-plex.
\item \textbf{SNAP\footnote{http://snap.stanford.edu/data/} and DIMACS10\footnote{https://www.cc.gatech.edu/dimacs10/downloads.shtml}}: This dataset contains 39 graphs with up to $1.05 \times 10^6$ vertices from the
Stanford large network dataset collection (SNAP) and the 10th DIMACS implementation challenge, which are used in both KDBB and MADEC$^+$~\cite{MADEC}. 
\item \textbf{DIMACS2\footnote{http://archive.dimacs.rutgers.edu/pub/challenge/graph/\\benchmarks/clique/}}: This dataset contains 49 dense graphs with up to 1,500 vertices from the 2nd DIMACS implementation challenge, which is used in MADEC$^+$.
\end{itemize}

For each graph, we generated six MDCP instances with $k = 1,3,5,10,15,20$ as KDBB did. Therefore, there are a total of $6 \times (114+139+39+49) = 2046$ MDCP instances.

\input{tabs/comp-results2}

\subsection{Solvers}

To evaluate the performance of our proposed KD-Club algorithm, we select the state-of-the-art BnB MDCP algorithms, MADEC$^+$~\cite{MADEC} and KDBB~\cite{KDBB} as the baseline algorithms. To evaluate the effectiveness of CLUB in different stages of KD-Club, we generate two variant algorithms of KD-Club, called KD-Club$_\text{Pre}^-$ and KD-Club$_\text{BnB}^-$. Details %of all the algorithms involved in our experiments 
are as follows.

%KD-Club was compared with the state-of-the-art BnB algorithms MADEC$^+$\footnote{https://github.com/chenxiaoyu233/k-defective}~\cite{MADEC} and KDBB~\cite{KDBB}. Since MADEC$^+$ fails to solve most of instances for $k>3$, for further comparison with its bound, we combine our preprocessing method with it as KDBB dose, the solver denotes $\text{MADEC}_\text{P}^+$. In addition to this, there are two variant of our solver, $\text{KDBB}_\text{P}^+$ and $\text{KD-Club}_\text{P}^-$, which will be used to verify the validity of our methods during the ablation study.
\begin{itemize}
\item \textbf{MADEC$^+$}: A BnB MDCP algorithm with a rough coloring-based upper bound that increases sharply with the increment of $k$. It shows good performance on instances based on dense DIMACS2 graphs with small $k$ values. The source code is available at\footnote{https://github.com/chenxiaoyu233/k-defective}.
\item \textbf{KDBB}: The best-performing BnB MDCP algorithm for instances based on massive sparse graphs and instances with large $k$ values. Since its code has not been open source, we use our implemented version in the experiments, which shows better performance than the results reported in the literature.
\item \textbf{KD-Club}: An implementation of our algorithm\footnote{The source codes of KD-Club are available at https://github.com/JHL-HUST/KD-Club}.
\item \textbf{KD-Club$_\text{Pre}^-$}: A variant of KD-Club without the CLUB during preprocessing, \ie, replacing the preprocessing in KD-Club with that in KDBB; Also a variant of KDBB with the BnB searching process in KD-Club.
\item \textbf{KD-Club$_\text{BnB}^-$}: A variant of KD-Club without the CLUB during BnB searching, \ie, replacing the BnB searching process in KD-Club with that in KDBB; Also a variant of KDBB with the preprocessing method in KD-Club.
%\item $\text{MADEC$^+$}_\text{P}^+$: An enhanced solver of MADEC$^+$ incorporating our preprocessing algorithm.
%\item $\text{KDBB}_\text{P+}$: A variant of KDBB that uses our preprocessing method to replace its original one. This solver is mainly used to verify the effectiveness of CLUB during pruning branches.
%\item $\text{KD-Club}_\text{P}^-$: A variant of KD-Club, using the original KDBB preprocessing method, mainly used to validate the effectiveness of our CLUB during preprocessing.
\end{itemize}

\subsection{Performance Comparison}

We first compare KD-Club, KDBB, and MADEC$^+$ on all four benchmarks to evaluate the overall performance of these solvers. The results are summarized in Table~\ref{tab-comp-results}, which presents the number of instances that can be solved within the cut-off time by each algorithm in solving instances of each benchmark grouped according to the $k$ values.

From the results, one can see that KD-Club solves considerably more instances than the baselines on most groups of benchmark instances, especially on instances with larger $k$ values. Since the upper bound in MADEC$^+$ increases significantly with the increment of $k$, it fails to solve most instances with $k > 5$. With the increment of $k$, the reduction rules in KDBB can hardly reduce the graph. Thus, its performance also declines significantly. 
With the benefit of our CLUB that considers the missing edges more thoroughly, the preprocessing and BnB stages in KD-Club are both very efficient, and thus KD-Club shows excellent performance on various benchmarks, even with large $k$ values. Moreover, although the number of solved instances of KD-Club is not much larger than the baselines in solving instances with small $k$ values, the running time of KD-Club to solve the instances is usually much shorter than that of the baselines, as indicated by the follow-up experiments.

We further present the detailed results of KD-Club and KDBB in solving 20 representative instances from four benchmarks with $k = 3$ and $k = 10$, as shown in Table~\ref{tab-detailed-result}. The results include the number of vertices (column $|V|$) and edges (column $|E|$) of each original graph, the number of vertices (column $|V'|$) and edges (column $|E'|$) of each graph after reducing by the preprocessing method of each algorithm, the running time in seconds (column \textit{Time}) and the sizes of their entire search trees in $10^4$ (column \textit{Tree}) to solve the instances. The symbol `NA' means the algorithm cannot solve the instance within the cut-off time.

The results show that when solving massive sparse graphs, such as the Facebook instances started with `socfb', the preprocessing method based on CLUB can help KD-Club reduce the graph size to an order of magnitude smaller than the graph reduced by the preprocessing in KDBB, indicating a significant reduction. When solving dense graphs, such as the DIMACS2 instances C125-9, johnson8-4-4, and san200-0-7-1, the preprocessing cannot reduce any vertex or edge, while the BnB process based on CLUB can still help KD-Club solve these instances with much shorter running time as compared to the cut-off time, but KDBB cannot even solve them within the cut-off time. As a result, both the search tree sizes and running time of KD-Club are several orders of magnitude smaller than those of KDBB in solving these instances with either small or large $k$ values.

%Our first experiment is used to compare the overall performance gap between KD-Club and the two compared algorithms. The algorithms were tested on 341 instances with a cutoff time of 1800s for each instance. The results are detailed in Table~\ref{tab-comp-results} and are grouped by each value of $k=1,3,5,10,15,20$, respectively.

\input{tabs/detailed_results4}

\begin{figure*}[!t]
    \centering
    \subfigure[On instances with $k=3$] {
    \includegraphics[width=0.85\columnwidth]{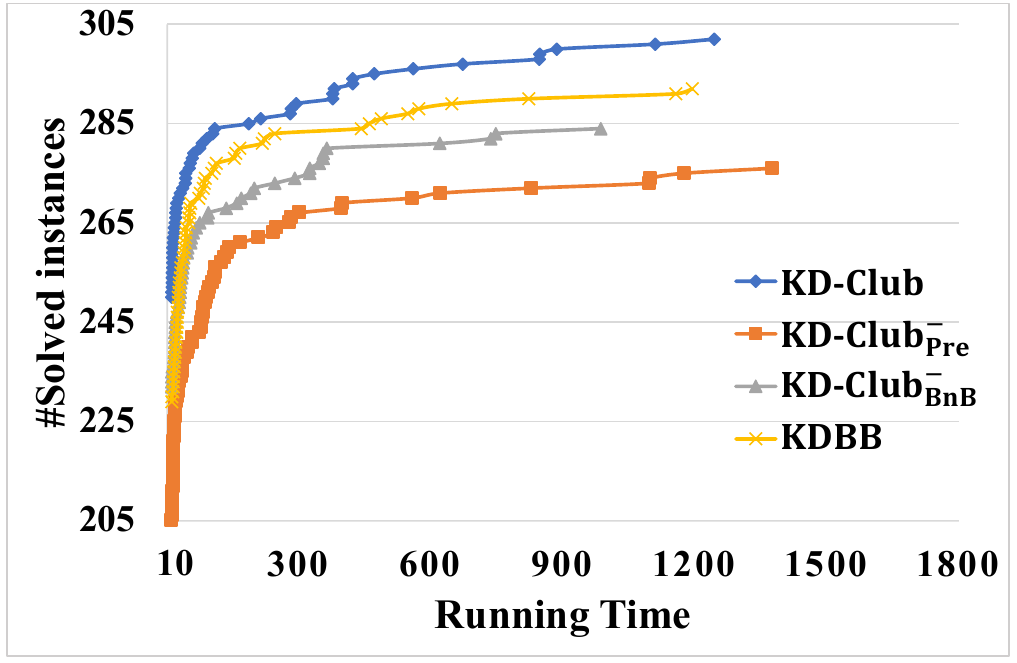}
    \label{fig:ablation3}}~~~~
    \subfigure[On instances with $k=10$]{
    \includegraphics[width=0.85\columnwidth]{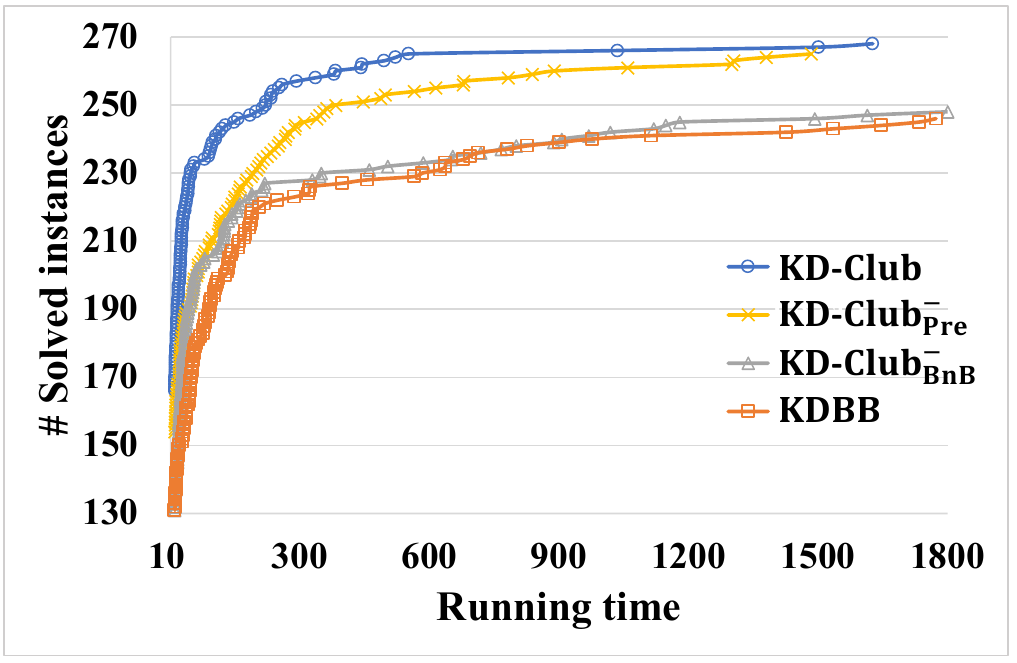}}
    \caption{Ablation study on MDCP instances over all the 341 graphs with 
    %$k = 3$ or $k = 10$.
    different $k$ values. 
    }
\label{fig-Ablation}
\end{figure*}

\subsection{Ablation Study}
%For ablation study, 
We then compare KD-Club with its two variants, KD-Club$_\text{Pre}^-$ and KD-Club$_\text{BnB}^-$, as well as KDBB to evaluate the effectiveness of CLUB in preprocessing and BnB searching stages. The results are shown in Figure~\ref{fig-Ablation}, where we present the variation of the number of solved instances with $k = 3$ and $k = 10$ for each algorithm over all the 341 graphs during the running time (in seconds). To present more clearly, we omit the results of easy instances solved within 10 seconds for each algorithm.

%From the results one can observe that 
In general, KD-Club yields better performance than the two variants, and the two variants perform better than KDBB, indicating that using CLUB in both preprocessing and BnB stages can improve the performance of the BnB algorithm. Moreover, when solving instances with $k = 10$, the performance of KD-Club and KD-Club$_\text{Pre}^-$ is similar, and the performance of KD-Club$_\text{BnB}^-$ and KDBB is similar. This is because the larger the value of $k$, the smaller the size of the graph that can be reduced by preprocessing, and the worse the effectiveness of preprocessing. On the other hand, when solving instances with large $k$ values, the BnB searching process based on CLUB also shows excellent performance, helping KD-Club significantly outperform KD-Club$_\text{BnB}^-$, and KD-Club$_\text{Pre}^-$ significantly outperform KDBB.

%% file: tabs/comp-results2.tex
% \usepackage{multirow}
\begin{table*}[!t]
\centering
\footnotesize
%\resizebox{\linewidth}{!}{
\begin{tabular}{l|ccc|ccc|ccc|ccc}\toprule
\multicolumn{1}{l|}{\multirow{2}{*}{$k$}} &\multicolumn{3}{c|}{Facebook} & \multicolumn{3}{c|}{Realworld} & \multicolumn{3}{c|}{SNAP and DIMACS10} & \multicolumn{3}{c}{DIMACS2} \\
\multicolumn{1}{c|}{}                   & KD-Club     & \hspace{-0.75em} KDBB    & \hspace{-0.75em}MADEC$^+$    & KD-Club     & \hspace{-0.75em}KDBB     & \hspace{-0.75em}MADEC$^+$    & KD-Club        & \hspace{-0.75em}KDBB       & \hspace{-0.75em}MADEC$^+$       & KD-Club    & \hspace{-0.75em}KDBB    & \hspace{-0.75em}MADEC$^+$   \\ \hline
1                                       & \textbf{112}      & \hspace{-0.75em} 110     & \hspace{-0.75em}12      & \textbf{130}      & \hspace{-0.75em}124      & \hspace{-0.75em}81      & \textbf{39}          & \hspace{-0.75em}\textbf{39}         & \hspace{-0.75em}24          & 31      & \hspace{-0.75em}17      & \hspace{-0.75em}\textbf{32}      \\
3                                       & \textbf{112}      & \hspace{-0.75em} 110     & \hspace{-0.75em}0      & \textbf{125}      & \hspace{-0.75em}116      & \hspace{-0.75em}62      & \textbf{38}          & \hspace{-0.75em}\textbf{38}         & \hspace{-0.75em}23          & \textbf{27}      & \hspace{-0.75em}12      & \hspace{-0.75em}14      \\
5                                       & \textbf{111}      & \hspace{-0.75em} 109     & \hspace{-0.75em}0      & \textbf{121}      & \hspace{-0.75em}110      & \hspace{-0.75em}52       & 37          & \hspace{-0.75em}\textbf{38}         & \hspace{-0.75em}23          & \textbf{25}      & \hspace{-0.75em}12      & \hspace{-0.75em}12      \\
10                                      & \textbf{109}      & \hspace{-0.75em} 108     & \hspace{-0.75em}0       & \textbf{106}      & \hspace{-0.75em}94       & \hspace{-0.75em}28       & \textbf{34}          & \hspace{-0.75em}\textbf{34}         & \hspace{-0.75em}11          & \textbf{19}      & \hspace{-0.75em}11      & \hspace{-0.75em}6       \\
15                                      & \textbf{108}      & \hspace{-0.75em} 104     & \hspace{-0.75em}0        & \textbf{94}       & \hspace{-0.75em}70       & \hspace{-0.75em}22       & \textbf{30}          & \hspace{-0.75em}28         & \hspace{-0.75em}9          & \textbf{13}      & \hspace{-0.75em}11      & \hspace{-0.75em}2       \\
20                                      & \textbf{104}      & \hspace{-0.75em} 86      & \hspace{-0.75em}0        & \textbf{85}       & \hspace{-0.75em}59       & \hspace{-0.75em}16       & \textbf{27}          & \hspace{-0.75em}24         & \hspace{-0.75em}6           & \textbf{12}      & \hspace{-0.75em}11      & \hspace{-0.75em}1      
\\\bottomrule
\end{tabular}
%}

\caption{Summarized results of KD-Club, KDBB, and MADEC$^+$ on four benchmarks. The best results appear in bold.}
\label{tab-comp-results}
\end{table*}

%\hspace{-0.75em}

%% file: tabs/detailed_results4.tex
% \usepackage{multirow}
\begin{table*}[!t]
    %\fontsize{9pt}{15}
    \footnotesize
    \centering
    % \resizebox{\linewidth}{!}{
    \begin{tabular}{lrr|rrrr:rrrr|rrrr:rrrr}\toprule
    \multicolumn{1}{l}{\multirow{3}{*}{Instance}} & \multicolumn{1}{r}{\multirow{3}{*}{\hspace{-1.1em}$|V|$}} & \multicolumn{1}{r|}{\multirow{3}{*}{\hspace{-1.1em}$|E|$}} & \multicolumn{8}{c|}{$k=3$}                                                                                                                                                                                           & \multicolumn{8}{c}{$k=10$}                                                                                                                                                                                          \\ \cline{4-19} 
    \multicolumn{1}{c}{}                           & \multicolumn{1}{c}{}                     & \multicolumn{1}{c|}{}                     & \multicolumn{4}{c:}{KD-Club}                                                                                & \multicolumn{4}{c|}{KDBB}                                                                                & \multicolumn{4}{c:}{KD-Club}                                                                                & \multicolumn{4}{c}{KDBB}                                                                                \\
    \multicolumn{1}{c}{}                           & \multicolumn{1}{c}{}                     & \multicolumn{1}{c|}{}                     & \multicolumn{1}{r}{\hspace{-0.5em}$|V'|$} & \multicolumn{1}{r}{\hspace{-1.1em}$|E'|$} & \multicolumn{1}{r}{\hspace{-1.1em}Tree} & \multicolumn{1}{r:}{\hspace{-1.1em}Time} & \multicolumn{1}{r}{\hspace{-0.5em}$|V'|$} & \multicolumn{1}{r}{\hspace{-1.1em}$|E'|$} & \multicolumn{1}{r}{\hspace{-1.1em}Tree} & \multicolumn{1}{r|}{\hspace{-1.1em}Time} & \multicolumn{1}{r}{\hspace{-0.5em}$|V'|$} & \multicolumn{1}{r}{\hspace{-1.1em}$|E'|$} & \multicolumn{1}{r}{\hspace{-1.1em}Tree} & \multicolumn{1}{r:}{\hspace{-1.1em}Time} & \multicolumn{1}{r}{\hspace{-0.5em}$|V'|$} & \multicolumn{1}{r}{\hspace{-1.1em}$|E'|$} & \multicolumn{1}{r}{\hspace{-1.1em}Tree} & \multicolumn{1}{r}{\hspace{-1.1em}Time} \\ \hline
    C125-9                                         & \hspace{-1.1em}125                                      & \hspace{-1.1em}6963                                      & \hspace{-0.5em}125                     & \hspace{-1.1em}6963                    & \hspace{-1.1em}\textbf{35.8}           & \hspace{-1.1em}\textbf{23.8}           & \hspace{-0.5em}125                     & \hspace{-1.1em}6963                    & \hspace{-1.1em}NA                       & \hspace{-1.1em}NA                        & \hspace{-0.5em}125                     & \hspace{-1.1em}6963                    & \hspace{-1.1em}\textbf{536}           & \hspace{-1.1em}\textbf{137}           & \hspace{-0.5em}125                     & \hspace{-1.1em}6963                    & \hspace{-1.1em}NA                       & \hspace{-1.1em}NA                       \\
    gen200-p0-9-55                                 & \hspace{-1.1em}200                                      & \hspace{-1.1em}17910                                     & \hspace{-0.5em}200                     & \hspace{-1.1em}17910                   & \hspace{-1.1em}\textbf{45.7}           & \hspace{-1.1em}\textbf{102}           & \hspace{-0.5em}200                     & \hspace{-1.1em}17910                   & \hspace{-1.1em}NA                       & \hspace{-1.1em}NA                        & \hspace{-0.5em}200                     & \hspace{-1.1em}17910                   & \hspace{-1.1em}\textbf{325}           & \hspace{-1.1em}\textbf{275}           & \hspace{-0.5em}200                     & \hspace{-1.1em}17910                   & \hspace{-1.1em}NA                       & \hspace{-1.1em}NA                       \\
    ia-wiki-Talk                                   & \hspace{-1.1em}92117                                    & \hspace{-1.1em}360767                                    & \hspace{-0.5em}179                     & \hspace{-1.1em}2244                    & \hspace{-1.1em}\textbf{0.04}           & \hspace{-1.1em}\textbf{2.95}           & \hspace{-0.5em}1119                    & \hspace{-1.1em}46194                   & \hspace{-1.1em}NA                       & \hspace{-1.1em}NA                        & \hspace{-0.5em}2492                    & \hspace{-1.1em}68463                   & \hspace{-1.1em}\textbf{334}           & \hspace{-1.1em}\textbf{1420}            & \hspace{-0.5em}9424                    & \hspace{-1.1em}155317                  & \hspace{-1.1em}NA                       & \hspace{-1.1em}NA                       \\
    johnson8-4-4                                   & \hspace{-1.1em}70                                       & \hspace{-1.1em}1855                                      & \hspace{-0.5em}70                      & \hspace{-1.1em}1855                    & \hspace{-1.1em}\textbf{3.72}           & \hspace{-1.1em}\textbf{0.96}           & \hspace{-0.5em}70                      & \hspace{-1.1em}1855                    & \hspace{-1.1em}184                    & \hspace{-1.1em}33.6                     & \hspace{-0.5em}70                      & \hspace{-1.1em}1855                    & \hspace{-1.1em}\textbf{3436}            & \hspace{-1.1em}\textbf{339}           & \hspace{-0.5em}70                      & \hspace{-1.1em}1855                    & \hspace{-1.1em}NA                       & \hspace{-1.1em}NA                       \\
    MANN-a27                                       & \hspace{-1.1em}378                                      & \hspace{-1.1em}70551                                     & \hspace{-0.5em}378                     & \hspace{-1.1em}70551                   & \hspace{-1.1em}\textbf{102}           & \hspace{-1.1em}\textbf{460}           & \hspace{-0.5em}378                     & \hspace{-1.1em}70551                   & \hspace{-1.1em}NA                       & \hspace{-1.1em}NA                        & \hspace{-0.5em}378                     & \hspace{-1.1em}70551                   & \hspace{-1.1em}\textbf{100}           & \hspace{-1.1em}\textbf{469}           & \hspace{-0.5em}378                     & \hspace{-1.1em}70551                   & \hspace{-1.1em}NA                       & \hspace{-1.1em}NA                       \\
    san200-0-7-1                                   & \hspace{-1.1em}200                                      & \hspace{-1.1em}13930                                     & \hspace{-0.5em}200                     & \hspace{-1.1em}13930                   & \hspace{-1.1em}\textbf{0.87}           & \hspace{-1.1em}\textbf{2.79}           & \hspace{-0.5em}200                     & \hspace{-1.1em}13930                   & \hspace{-1.1em}NA                       & \hspace{-1.1em}NA                        & \hspace{-0.5em}200                     & \hspace{-1.1em}13930                   & \hspace{-1.1em}\textbf{36.9}           & \hspace{-1.1em}\textbf{34.2}           & \hspace{-0.5em}200                     & \hspace{-1.1em}13930                   & \hspace{-1.1em}NA                       & \hspace{-1.1em}NA                       \\
    %scc\_rt\_lebanon                               & \hspace{-1.1em}3370                                     & \hspace{-1.1em}5                                         & \hspace{-1.1em}3370                    & \hspace{-1.1em}5                       & \hspace{-1.1em}\textbf{0.00}           & \hspace{-1.1em}\textbf{1.17}           & \hspace{-1.1em}3370                    & \hspace{-1.1em}5                       & \hspace{-1.1em}1135                     & \hspace{-1.1em}169                     & \hspace{-1.1em}3370                    & \hspace{-1.1em}5                       & \hspace{-1.1em}\textbf{0.01}           & \hspace{-1.1em}\textbf{1.18}           & \hspace{-1.1em}3370                    & \hspace{-1.1em}5                       & \hspace{-1.1em}NA                       & \hspace{-1.1em}NA                       \\
    Slashdot0811                                   & \hspace{-1.1em}77360                                    & \hspace{-1.1em}469180                                    & \hspace{-0.5em}138                     & \hspace{-1.1em}5411                    & \hspace{-1.1em}\textbf{0.14}           & \hspace{-1.1em}\textbf{0.65}           & \hspace{-0.5em}246                     & \hspace{-1.1em}9888                    & \hspace{-1.1em}65.7                    & \hspace{-1.1em}80.9                     & \hspace{-0.5em}210                     & \hspace{-1.1em}7787                    & \hspace{-1.1em}\textbf{15.3}           & \hspace{-1.1em}\textbf{10.0}           & \hspace{-0.5em}503                     & \hspace{-1.1em}16319                   & \hspace{-1.1em}428                    & \hspace{-1.1em}75.8                    \\
    soc-digg                                       & \hspace{-1.1em}770799                                   & \hspace{-1.1em}5907132                                   & \hspace{-0.5em}199                     & \hspace{-1.1em}13551                   & \hspace{-1.1em}\textbf{0.61}           & \hspace{-1.1em}\textbf{70.0}           & \hspace{-0.5em}5580                    & \hspace{-1.1em}875610                  & \hspace{-1.1em}NA                       & \hspace{-1.1em}NA                        & \hspace{-0.5em}519                     & \hspace{-1.1em}37458                   & \hspace{-1.1em}\textbf{29.3}           & \hspace{-1.1em}\textbf{214}           & \hspace{-0.5em}7349                    & \hspace{-1.1em}1141673                 & \hspace{-1.1em}NA                       & \hspace{-1.1em}NA                       \\
    socfb-Cornell5                                 & \hspace{-1.1em}18660                                    & \hspace{-1.1em}790777                                    & \hspace{-0.5em}211                     & \hspace{-1.1em}6552                    & \hspace{-1.1em}\textbf{0.06}           & \hspace{-1.1em}\textbf{4.45}           & \hspace{-0.5em}1292                    & \hspace{-1.1em}52540                   & \hspace{-1.1em}22.6                    & \hspace{-1.1em}89.2                     & \hspace{-0.5em}541                     & \hspace{-1.1em}15848                   & \hspace{-1.1em}\textbf{0.83}           & \hspace{-1.1em}\textbf{22.6}           & \hspace{-0.5em}2174                    & \hspace{-1.1em}84381                   & \hspace{-1.1em}44.7                    & \hspace{-1.1em}107                    \\
    socfb-Indiana                                  & \hspace{-1.1em}29732                                    & \hspace{-1.1em}1305757                                   & \hspace{-0.5em}136                     & \hspace{-1.1em}4394                    & \hspace{-1.1em}\textbf{0.04}           & \hspace{-1.1em}\textbf{7.92}           & \hspace{-0.5em}1596                    & \hspace{-1.1em}57741                   & \hspace{-1.1em}2.84                    & \hspace{-1.1em}22.2                     & \hspace{-0.5em}844                     & \hspace{-1.1em}29290                   & \hspace{-1.1em}\textbf{1.07}           & \hspace{-1.1em}\textbf{17.5}           & \hspace{-0.5em}2884                    & \hspace{-1.1em}104681                  & \hspace{-1.1em}16.8                    & \hspace{-1.1em}152                    \\
    socfb-OR                                       & \hspace{-1.1em}63392                                    & \hspace{-1.1em}816886                                    & \hspace{-0.5em}152                     & \hspace{-1.1em}2901                    & \hspace{-1.1em}\textbf{0.08}           & \hspace{-1.1em}\textbf{2.09}           & \hspace{-0.5em}1440                    & \hspace{-1.1em}35123                   & \hspace{-1.1em}5.21                    & \hspace{-1.1em}12.0                     & \hspace{-0.5em}621                     & \hspace{-1.1em}14999                   & \hspace{-1.1em}\textbf{2.65}           & \hspace{-1.1em}\textbf{23.5}           & \hspace{-0.5em}4910                    & \hspace{-1.1em}117554                  & \hspace{-1.1em}42.2                    & \hspace{-1.1em}511                    \\
    socfb-Penn94                                   & \hspace{-1.1em}41536                                    & \hspace{-1.1em}1362220                                   & \hspace{-0.5em}71                      & \hspace{-1.1em}2228                    & \hspace{-1.1em}\textbf{0.04}           & \hspace{-1.1em}\textbf{4.69}           & \hspace{-0.5em}875                     & \hspace{-1.1em}23454                   & \hspace{-1.1em}1.07                    & \hspace{-1.1em}4.86                     & \hspace{-0.5em}237                     & \hspace{-1.1em}6564                    & \hspace{-1.1em}\textbf{0.42}           & \hspace{-1.1em}\textbf{7.03}           & \hspace{-0.5em}2087                    & \hspace{-1.1em}58201                   & \hspace{-1.1em}2.86                    & \hspace{-1.1em}34.3                    \\
    socfb-Rice31                                   & \hspace{-1.1em}4087                                     & \hspace{-1.1em}184828                                    & \hspace{-0.5em}83                      & \hspace{-1.1em}1601                    & \hspace{-1.1em}\textbf{0.03}           & \hspace{-1.1em}\textbf{2.36}           & \hspace{-0.5em}1170                    & \hspace{-1.1em}36409                   & \hspace{-1.1em}10.9                    & \hspace{-1.1em}16.9                     & \hspace{-0.5em}921                     & \hspace{-1.1em}25731                   & \hspace{-1.1em}\textbf{2.34}           & \hspace{-1.1em}\textbf{35.2}           & \hspace{-0.5em}2161                    & \hspace{-1.1em}87191                   & \hspace{-1.1em}80.3                    & \hspace{-1.1em}142                    \\
    %socfb-Tennessee95                              & \hspace{-1.1em}16979                                    & \hspace{-1.1em}770659                                    & \hspace{-0.5em}87                      & \hspace{-1.1em}3604                    & \hspace{-1.1em}\textbf{0.06}           & \hspace{-1.1em}\textbf{4.84}           & \hspace{-0.5em}802                     & \hspace{-1.1em}31594                   & \hspace{-1.1em}33.0                    & \hspace{-1.1em}19.2                     & \hspace{-0.5em}199                     & \hspace{-1.1em}8737                    & \hspace{-1.1em}\textbf{0.64}           & \hspace{-1.1em}\textbf{10.0}           & \hspace{-0.5em}1352                    & \hspace{-1.1em}55156                   & \hspace{-1.1em}72.4                    & \hspace{-1.1em}36.9                    \\
    socfb-Texas84                                  & \hspace{-1.1em}36364                                    & \hspace{-1.1em}1590651                                   & \hspace{-0.5em}129                     & \hspace{-1.1em}6606                    & \hspace{-1.1em}\textbf{0.15}           & \hspace{-1.1em}\textbf{8.99}           & \hspace{-0.5em}927                     & \hspace{-1.1em}41536                   & \hspace{-1.1em}85.4                    & \hspace{-1.1em}101                     & \hspace{-0.5em}461                     & \hspace{-1.1em}18985                   & \hspace{-1.1em}\textbf{17.6}           & \hspace{-1.1em}\textbf{40.7}           & \hspace{-0.5em}1494                    & \hspace{-1.1em}66131                   & \hspace{-1.1em}601                    & \hspace{-1.1em}399                    \\
    socfb-UF                                       & \hspace{-1.1em}35111                                    & \hspace{-1.1em}1465654                                   & \hspace{-0.5em}299                     & \hspace{-1.1em}10983                   & \hspace{-1.1em}\textbf{0.06}           & \hspace{-1.1em}\textbf{15.2}           & \hspace{-0.5em}1612                    & \hspace{-1.1em}75851                   & \hspace{-1.1em}54.7                    & \hspace{-1.1em}82.7                     & \hspace{-0.5em}975                     & \hspace{-1.1em}42128                   & \hspace{-1.1em}\textbf{5.51}           & \hspace{-1.1em}\textbf{111}           & \hspace{-0.5em}2088                    & \hspace{-1.1em}98361                   & \hspace{-1.1em}237                    & \hspace{-1.1em}159                    \\
    socfb-UGA50                                    & \hspace{-1.1em}24389                                    & \hspace{-1.1em}1174057                                   & \hspace{-0.5em}125                     & \hspace{-1.1em}6457                    & \hspace{-1.1em}\textbf{0.10}           & \hspace{-1.1em}\textbf{7.87}           & \hspace{-0.5em}1276                    & \hspace{-1.1em}55310                   & \hspace{-1.1em}80.3                    & \hspace{-1.1em}103                     & \hspace{-0.5em}457                     & \hspace{-1.1em}18956                   & \hspace{-1.1em}\textbf{11.4}           & \hspace{-1.1em}\textbf{39.8}           & \hspace{-0.5em}2093                    & \hspace{-1.1em}89760                   & \hspace{-1.1em}472                    & \hspace{-1.1em}286                    \\
    socfb-Uillinois                                & \hspace{-1.1em}30795                                    & \hspace{-1.1em}1264421                                   & \hspace{-0.5em}104                     & \hspace{-1.1em}4828                    & \hspace{-1.1em}\textbf{0.06}           & \hspace{-1.1em}\textbf{6.58}           & \hspace{-0.5em}586                     & \hspace{-1.1em}25096                   & \hspace{-1.1em}14.0                    & \hspace{-1.1em}13.1                     & \hspace{-0.5em}232                     & \hspace{-1.1em}11219                   & \hspace{-1.1em}\textbf{1.20}           & \hspace{-1.1em}\textbf{8.37}           & \hspace{-0.5em}1028                    & \hspace{-1.1em}43222                   & \hspace{-1.1em}76.2                    & \hspace{-1.1em}27.8                    \\
    socfb-Vassar85                                 & \hspace{-1.1em}3068                                     & \hspace{-1.1em}119161                                    & \hspace{-0.5em}80                      & \hspace{-1.1em}1220                    & \hspace{-1.1em}\textbf{0.02}           & \hspace{-1.1em}\textbf{0.65}           & \hspace{-0.5em}940                     & \hspace{-1.1em}22547                   & \hspace{-1.1em}0.55                    & \hspace{-1.1em}3.36                     & \hspace{-0.5em}496                     & \hspace{-1.1em}9340                    & \hspace{-1.1em}\textbf{1.40}           & \hspace{-1.1em}\textbf{18.3}           & \hspace{-0.5em}2591                    & \hspace{-1.1em}99363                   & \hspace{-1.1em}7.66                    & \hspace{-1.1em}221                    \\
    socfb-Yale4                                    & \hspace{-1.1em}8578                                     & \hspace{-1.1em}405450                                    & \hspace{-0.5em}140                     & \hspace{-1.1em}3145                    & \hspace{-1.1em}\textbf{0.03}           & \hspace{-1.1em}\textbf{2.62}           & \hspace{-0.5em}1033                    & \hspace{-1.1em}28691                   & \hspace{-1.1em}4.35                    & \hspace{-1.1em}7.87                     & \hspace{-0.5em}809                     & \hspace{-1.1em}20248                   & \hspace{-1.1em}\textbf{1.19}           & \hspace{-1.1em}\textbf{31.4}           & \hspace{-0.5em}2883                    & \hspace{-1.1em}95160                   & \hspace{-1.1em}31.4                    & \hspace{-1.1em}161                    \\
    Wiki-Vote                                      & \hspace{-1.1em}7115                                     & \hspace{-1.1em}100762                                    & \hspace{-0.5em}85                      & \hspace{-1.1em}1602                    & \hspace{-1.1em}\textbf{0.10}           & \hspace{-1.1em}\textbf{1.39}           & \hspace{-0.5em}902                     & \hspace{-1.1em}34228                   & \hspace{-1.1em}182                    & \hspace{-1.1em}431                     & \hspace{-0.5em}1378                    & \hspace{-1.1em}40535                   & \hspace{-1.1em}\textbf{40.4}           & \hspace{-1.1em}\textbf{171}           & \hspace{-0.5em}2233                    & \hspace{-1.1em}72679                   & \hspace{-1.1em}NA                       & \hspace{-1.1em}NA            \\ \bottomrule          
    \end{tabular}
    % }
    \caption{Detailed results of KD-Club and KDBB on 20 representative MDCP instances from four benchmarks with $k = 3$ and $k=10$. The search tree size is in $10^4$, and the time is in seconds. The better results appear in bold.}
    \label{tab-detailed-result}
    \end{table*}

%% file: body/06-Con.tex
For the NP-hard Maximum $k$-Defective Clique Problem (MDCP), we proposed a novel CoLoring-based Upper Bound (CLUB) and then a new BnB algorithm called KD-Club. CLUB considers the missing edges that are ignored by previous methods, and uses graph coloring techniques in an ingenious way to calculate lower bounds on the number of missing edges, so as to obtain a tighter upper bound. % of the size of the maximum $k$-defective clique that can be extended from the current partial solution. 
By using CLUB %to reduce the graph and prune the branches
on graph reduction and branch pruning, KD-Club significantly outperforms state-of-the-art BnB MDCP algorithms and exhibits excellent performance and robustness on instances based on either massive sparse or dense graphs, with either small or large $k$ values.

In future work, we plan to apply our approach to improve the upper bounds used in BnB algorithms for other combinatorial optimization problems, such as %problems related to other relaxation clique models. 
other relaxation clique problems. 
In fact, existing BnB algorithms might ignore part of the relationship between elements in the problem to be solved, and could be greatly improved via a similar relationship detection method.

%The advantage of CLUB is that it not only considers the connectivity between the partial solution and the candidate set, but also fully considers the connectivity between the candidate vertices. Extensive experimental results show that the proposed KD-Club outperforms existing algorithms on both large sparse graphs and dense graphs. In future work, we plan to apply our method to other relaxed cliques problems.